\documentclass[a4paper,11pt]{article}
\usepackage{fullpage,amsthm,color}
\usepackage{titlesec}

\definecolor{myurlcolor}{rgb}{0,0,0.4}
\definecolor{mycitecolor}{rgb}{0,0.5,0}
\definecolor{myrefcolor}{rgb}{0.5,0,0}
\usepackage{hyperref}
\hypersetup{colorlinks,
linkcolor=myrefcolor,
citecolor=mycitecolor,
urlcolor=myurlcolor}

\usepackage[draft]{fixme}
\usepackage{tikz}
\usepackage{subfigure}
\usepackage{multicol}
\usepackage{amsmath,bbm}
\usepackage{graphicx}
\usepackage{amsfonts}
\usepackage{amssymb}
\usepackage{amsmath, amssymb, amsthm,verbatim,graphicx,bbm}
\usepackage{mathrsfs}
\usepackage{color,xcolor,longtable}
\usepackage{changes}
\newcommand{\beq}[0]{\begin{equation}}
\newcommand{\eeq}[0]{\end{equation}}
\newcommand{\bw}[0]{\begin{widetext}}
\newcommand{\ew}[0]{\end{widetext}}
\newcommand{\bc}[0]{\begin{center}}
\newcommand{\ec}[0]{\end{center}}
\newcommand{\bwn}[0]{\begin{widetext}\begin{eqnarray}}
\newcommand{\ewn}[0]{\end{eqnarray}\end{widetext}}
\newcommand{\beqn}[0]{\begin{eqnarray}}
\newcommand{\eeqn}[0]{\end{eqnarray}}

\newcommand{\np}[0]{{\it e.g.}}

\newcommand{\proj}[1]{|#1\rangle\!\langle #1|}
\newcommand{\ket}[1]{|#1\rangle}

\newcommand{\tr}[0]{\mathrm{tr}}

\newcommand{\jedynka}[0]{\mathbbm{1}}

\newcommand{\non}[0]{\nonumber\\}

\def\calB{{\cal B}}

\def\calH{{\cal H}}

\def\calP{{\cal P}}

\def\calS{{\cal S}}

\def\calV{{\cal V}}

\usepackage{mathpazo,palatino} 
\usetikzlibrary{patterns}

\newcommand{\one}{\leavevmode\hbox{\small1\normalsize\kern-.33em1}}
\def\tr{\mbox{tr}}

\def\beq{\begin{equation}}
\def\eeq{\end{equation}}
\def\be{\begin{equation}}
\def\ee{\end{equation}}
\def\ben{\begin{eqnarray}}
\def\een{\end{eqnarray}}
\def\eea{\end{array}}
\def\bea{\begin{array}}
\newcommand{\ot}[0]{\otimes}
\newcommand{\Tr}[1]{\mathrm{Tr}#1}
\newcommand{\bei}{\begin{itemize}}
\newcommand{\eei}{\end{itemize}}

\newcommand{\symm}[1]{P_{#1}^{\mathrm{sym}}}

\newcommand{\ges}[0]{P^{\mathrm{ges}}}

\newcommand{\stan}[2]{\varrho_{#1}^{#2}}

\renewcommand{\emph}[1]{\textbf{#1}}

\theoremstyle{plain}
\newtheorem{fakt}{Fact}
\newtheorem{thm}{Theorem}
\newtheorem*{thm*}{Theorem}
\newtheorem{lem}{Lemma}

\newtheorem{cor}[thm]{Corollary}

\theoremstyle{definition}

\theoremstyle{remark}

\titleformat*{\section}{\large\bfseries}

\begin{document}

\title{Constructions of genuinely entangled multipartite states with applications to local hidden variables (LHV) and states (LHS) models}

\author{R. Augusiak$^{1}$, M. Demianowicz$^{2}$, J. Tura$^{3}$
\\[0.5em]
{\it\small $^1$Center for Theoretical Physics, Polish Academy of Sciences, 02-668 Warszawa, Poland}\\
{\it\small $^2$Atomic and Optical Physics Division, Department of Atomic, Molecular and Optical Physics,}\\
{\it\small Faculty of Applied Physics and Mathematics,}\\
{\it\small Gda\'nsk University of Technology, Narutowicza 11/12, 80–-233 Gda\'nsk, Poland}\\
{\it\small $^3$Max-Planck-Institut f\"ur Quantenoptik, D-85748 Garching, Germany}\\
}
\date{}

\maketitle

\begin{abstract}
Building upon the results of [R. Augusiak \textit{et al.}, Phys. Rev. Lett. \textbf{115}, 030404 (2015)] we develop a general approach to the generation of genuinely entangled multipartite states of any number of parties from genuinely entangled states of a fixed number of parties, in particular, the bipartite entangled ones. In our approach, certain isometries whose output subspaces
are either symmetric or genuinely entangled in some multipartite Hilbert spaces are applied to local subsystems of bipartite entangled or multipartite genuinely entangled quantum states. To prove that entanglement of the resulting states is indeed genuine we then introduce  
novel criteria allowing to decide it efficiently.
The construction is then exploited to provide examples of 
multipartite states that are genuinely entangled but not genuinely nonlocal, giving further illustration for the
inequivalence between entanglement and nonlocality in the multiparticle scenario. It is also shown how to construct genuinely entangled states which are unsteerable across certain bipartite cuts.
\end{abstract}

\section{Introduction}

Quantum entanglement is a fascinating feature of composite physical systems. It not only 
shows that quantum physics drastically departs from classical physics, but, over the years, 
has also been turned into a key resource for a whole range of applications such as 
quantum teleportation \cite{teleport} or quantum cryptography \cite{crypto}. It may also give rise to yet another, stronger, feature of composite quantum systems, which is Bell nonlocality \cite{BellNon}. Thus, detection 
and characterization of entanglement in composite quantum systems remains the central, but still unsolved, problem in quantum information theory \cite{GezaOtfriedReview}.

With the development of experimental methods of generating, 
controlling, and manipulating quantum states consisting of 
more than two particles (see, e.g., Ref. \cite{Klempt,Basel}), the problem of entanglement detection and 
characterization has recently gained much importance. In comparison 
to the bipartite case, here this problem 
complicates significantly due to the fact that 
the multiparty scenario supports the whole 
variety of different types of entanglement (see Ref. \cite{DurCirac,AcinPRL2001,TothGuhneNJP2005} for various notions
developed to grasp the richness of entanglement in the multiparty scenario).
It turns out that of the broad variety of types of entanglement the most desirable one
from the application point of view is the so-called genuine multipartite entanglement \cite{GezaMetro2012,HyllusMetro2012,AugusiakMetro2016}. Roughly speaking, 
a state characterized by this type of entanglement is one
in which all particles all entangled with each other; more formally, according to the operational definitions of separability, it is a state which cannot be created from the scratch only using local operations and classical communication by spatially separated parties (so--called LOCC paradigm). In the recent
years we have thus witnessed  a considerable effort towards providing theoretical constructions of 
genuinely entangled states as well as experimental realizations of various interesting states
(see, \np, Ref. \cite{Klempt}).

Our aim with this contribution is to shed light on this  complex problem from
a slightly 
different angle. Building upon our recent work \cite{krotka}, 
we provide a general method for constructing 
genuinely entangled $N$-partite states with any $N$ from 
$K$-partite genuinely entangled states with $K<N$. In the particular, and at the same time the most interesting, case of $K=2$ this method can be used to embed the known classes of entangled bipartite states to the multiparty case in such a way that the resulting state
is genuinely entangled. Our construction consists in applying certain isometries
whose output subspaces are either symmetric or genuinely entangled in multipartite 
Hilbert spaces to the local subsystems of the initial bipartite or $K$-partite density matrices. A key ingredient of our approach are certain entanglement criteria that we derive here which allow one to check whether a given multipartite state is genuinely 
entangled. 

Then, following Ref. \cite{krotka}, we
show how this construction can be harnessed to obtain examples of genuinely entangled $N$-partite states that are not genuinely nonlocal with respect to the Svetlichny and the recent operational definitions of nonlocality \cite{NonlocalityDefinition,NonlocalityDefinition2}.
We thus obtain further examples of multipartite states illustrating the statement
made in Ref. \cite{krotka} that entanglement and nonlocality 
are inequivalent notions in the multipartite scenario (see also Ref. \cite{Hirsch} for GME states with fully local models). Finally, we show 
that these concepts can also be applied to the notion of steering in the 
multipartite case. In particular, we present examples of genuinely entangled states which are unsteerable across certain bipartitions.

The paper is organized as follows. In Sec. \ref{Sec:prelim} we
introduce some background information that is used throughout
the whole paper. In Sec. \ref{Sec:constr} we present our construction
along with entanglement conditions and examples of genuinely entangled
states obtained from the construction. In Sec. \ref{Sec:GMN} we show in 
a detailed way how our approach can be used to design genuinely entangled
states for any number of parties that are not genuinely nonlocal. We then 
discuss an application of these concepts to steering in the multipartite
scenario. We conclude in Section \ref{conclusion}.

\section{Preliminaries } 
\label{Sec:prelim}
This section sets up the scenario and introduces the relevant
notation and terminology from the area of multipartite
entanglement. The notions of (non)locality and (un)steerability in the relevant multipartite scenarios will be discussed in Section \ref{Sec:GMN}.

Consider $N$-parties $A_1,\ldots,A_N:=\bold{A}$ [for small $N$ they will be denoted $A$ (Alice), $B$ (Bob), etc.] holding an $N$--partite quantum state described by the density matrix $\rho_{\bold{A}}$ acting on $\mathcal{H}_{d,N}:=(\mathbb{C}^d)^{\otimes N}$, i.e., $\rho_{\bold{A}}\in \mathcal{B} \left( (\mathbb{C}^d)^{\otimes N} \right)$, where $\mathcal{B}(\mathcal{H})$ denotes the set of bounded linear operators acting on $\mathcal{H}$. 

 A division of the parties into $K$ nonempty disjoint groups $S_i$ (the number of elements in $S_i$ will be denoted by $|S_i|$) such that $S_1\cup\ldots\cup S_K=\bold{A}$ is called a {\it $K$--partition}, denoted by $S_1|S_2|\cdots|S_K$ or simply $\calS_K$; when $K=2$ such division is called a {\it bipartition } and denoted simply by $S|\bar{S}$, where $\bar{S}$ is the complement of $S$ in $\bold{A}$. The set of all $K$--partitions will be denoted by $\mathdollar_K$. 
 
  Any
$K$-partition induces a natural $K$-partition of $\mathcal{H}_{d,N}$ into the
corresponding subsystems, i.e., $\calH=\bigotimes _{p=1}^K \calH_{S_p}$ with $\calH_{S_p}$ being the Hilbert space corresponding to the parties from the set $S_p$. 
%
%

Now, a state is called {\it $K$--separable}
%
%
%
with respect to some $K$-partition $S_1|S_2|\cdots|S_K$ if it is a convex combination 
of pure states from $\mathcal{H}_{d,N}$ that are product with respect to this
$K$-partition, i.e., 
\begin{equation}\label{separable}
\rho_{\calS_K}=\sum_{i}p_i\proj{\psi_{S_1}^i}\ot\proj{\psi_{S_2}^i}\ot\ldots\ot
\proj{\psi_{S_K}^i},
\end{equation}
where $\ket{\psi_{S_j}^i}$ are pure states belonging to the
Hilbert space corresponding to the group $S_j$, i.e., $\ket{\psi_{S_j}^i} \in \calH_{S_j}$ . Then, we simply call an $N$--partite state
$\rho_{\mathbf{A}}$   \textit{$K$-separable}  if it is a convex combination
of quantum states that are $K$-separable across various $K$-partitions, 
which mathematically can be expressed as
\begin{equation}\label{separableK}
\rho_{\mathbf{A}}=\sum_{\calS_K\in\mathdollar_K}q_{\calS_K} \rho_{\calS_K},
\end{equation}
where  $\{ q_{\calS_K}\}$ is a probability distribution, whereas $\rho_{\calS_K}$
are mixed states separable with respect to a particular $K$-partition ${\calS_K}$, i.e., 
they admit the decomposition (\ref{separable}). In the particular case of $K=2$ one speaks about the {\it biseparability} of states. Importantly, any $(K+1)$-separable state is also $K$-separable.

At one end of the above
classification there are states that are $N$-separable ($K=N$ case) -- we call
them \textit{fully separable} because no entanglement whatsoever is present
in such states. At the other end, one has \textit{genuine
multipartite entangled} (GME) states which are those that do not
decompose into any probabilistic mixture of states that admit any
form of separability, even biseparable ones. In short, GME state are those which are not biseparable.

In what follows particular attention will be devoted to 
certain
subspaces. 
The first one is the symmetric subspace of $\mathcal{H}_{d,N}$, denoted
$\mathrm{Sym}((\mathbb{C}^d)^{\ot N})$.
%
%
It consists of vectors that are invariant under permutation of any pair of parties
and its dimension is $\binom{N+d-1}{d-1}$, which for qubits simplifies to $N+1$.
%
%
%
It is worth mentioning that in the particular case of $\mathcal{H}_{2,N}$, $\mathrm{Sym}((\mathbb{C}^2)^{\ot N})$ is spanned by 
the so-called symmetric Dicke states $\ket{D_{N,k}}$ which are normalized symmetrizations of
the simple kets $\ket{\{N,k\}}$ consisting of $k$ ones and $N-k$ zeros, i.e.,
\begin{equation}
\ket{D_{N,k}}=\frac{1}{\sqrt{\binom{N}{k}}}\,\mathrm{Symm}(\ket{0\ldots 01\ldots 1}),
\end{equation}
where $\mathrm{Symm}$ stands for the said symmetrization. For further benefits, let us also notice that $\mathrm{Sym}((\mathbb{C}^2)^{\ot N})$
contains either fully product vectors $\ket{e}^{\otimes N}$ with $\ket{e}\in\mathbbm{C}^2$
or genuinely entangled ones.  A word  about the terminology regarding symmetric subspaces used in the rest of the paper is in order: when we say that a subspace is symmetric we either mean that it is the symmetric subspace,  $\mathrm{Sym}((\mathbb{C}^d)^{\ot N})$, itself or it is simply a subspace of the latter.
%
%

We will then consider subspaces consisting solely of multiparty states that are genuinely multipartite entangled. Such subspaces, which can naturally be called \textit{genuinely entangled subspaces } (GESes), were  first mentioned in Ref.
\cite{Cubitt}, however, since then, they have been barely studied in the literature. 
A well-known example of a genuinely entangled subspace is
the antisymmetric subspace in $\mathcal{H}_{d,N}$ for $d\geq N$. It is spanned by entangled vectors that acquire the minus sign under permutation of any pair of parties and its dimension
is $\binom{d}{N}$. A recent effort towards characterization of such genuinely entangled subspaces has been reported in Ref. \cite{GESy}. 
Let us finally notice that GESes belong to a larger class of subspaces
called \textit{completely entangled subspaces} (CESes) \cite{Cubitt,CES1,ces-bhat,CES2}, which only contain  entangled vectors, however, not necessarily genuinely entangled. Consequently, a GES is also a CES, however, the opposite implication does not hold in general.

Projections onto the subspaces discussed above will be denoted by $\symm{S_i}$
%
%
and, $\ges_{S_i}$, accordingly, with a subscript indicating subsystems they act on.

\section{Generating $N$-partite GME states from $K$-partite GME states with
$K<N$}
\label{Sec:constr}
The aim of this section is to introduce, building upon the results of Ref. \cite{krotka}, 
a simple method of generating $N$-partite genuinely entangled states   from $K$-partite GME states with $K<N$. In the particular case of $K=2$ the method allows for an
extension of bipartite entangled states to GME states of {\it any} number of
parties. In further parts, this method will be employed to construct new examples of genuinely entangled multiparty states that are not genuinely multiparty nonlocal, and also use it to provide examples of GME states that are unsteerable with respect to certain cuts. 

\subsection{Entanglement conditions}

Our construction of GME states  relies heavily on entanglement conditions, which are generalizations of the criterion formulated in \cite{krotka}. For completeness, we recall it along with its proof below. The latter requires a lemma about a certain property of product vectors with symmetric subsystems. The proof is given in  Appendix A.
\begin{lem}\label{LLem1} \cite{krotka}
 Let $\ket{\psi}\in\mathcal{H}_{d,N}$ be a pure state product with respect to
some bipartition $T|\bar{T}$. If $P_{S}^{\mathrm{sym}}\ket{\psi}=\ket{\psi}$
where $S$ is a subset of $\textbf{A}$ having  nontrivial 
overlaps with $T$ and $\bar{T}$, i.e.,  $S\cap T \ne \emptyset$ and  $S\cap \bar{T} \ne \emptyset$, then $\ket{\psi}$ is also product with respect
to the bipartition $S|\bar{S}$.
\end{lem}

We can now give the announced condition from \cite{krotka}.


%
%
%
%
%
%

\begin{fakt}\label{fact1} \cite{krotka}
Let $\rho_{\bold{A}}$ be an $N$-partite state acting on $\mathcal{H}_{d,N}$
such that with respect to some $K$-partition
$S_1|S_2|\cdots|S_K$, its subsystems corresponding to $S_k$'s are
supported on the symmetric subspaces, i.e., 
\beq\label{Ksym}
\symm{S_1}\ot\cdots\ot\symm{S_K}\rho_{\bold{A}}\symm{S_1}\ot\cdots\ot\symm{S_K}=\rho_{\bold{A}}.
\eeq
Then, if $\rho_{\bold{A}}$ is not
GME, it  takes the biseparable form 
\begin{equation}
\label{dekompozycja}
 \rho_{\bold{A}}=\sum_{T|\bar{T}}p_{T|\bar{T}}\rho_{T|\bar{T}}^{\mathrm{sep}},
\end{equation}
where the sum runs over all bipartitions $T|\bar{T}$ for which  $T$ and $\bar{T}$ are unions of the sets $S_k$ and $\rho_{T|\bar{T}}^{\mathrm{sep}}$'s are separable across $T|\bar{T}$.
\end{fakt}
\begin{proof}
%
%
Since $\rho$ is not GME it admits the biseparable ($K=2$) decomposition (\ref{separableK}).
Assume then that in this decomposition there is a biseparable state $\varrho_{T|\bar{T}}$
where $T$, and, in consequence, also $\bar{T}$, are \textit{not} unions of the sets $S_k$. Obviously, such
$\varrho_{T|\bar{T}}$  admits the decomposition (\ref{separable}) in which there only appear
pure states that are product with respect to $T|\bar{T}$, but not with respect
to sets that are unions of $S_k$, by the assumption above. 

Let $\ket{\psi_T}\ket{\phi_{\bar{T}}}$ be
any of such product states. Now, since $T$ is not a union of $S_k$'s, 
there exists at least one set among the latter whose overlap with both $T$ and $\bar{T}$ is
nonempty. 
Further, for at least one such set, say $S_{i_m}$, 
either the state $\ket{\psi_T}$ is not product with respect to 
the bipartition $(T\cap S_{i_m})|[T\setminus (T\cap S_{i_m})]$ or 
$\ket{\phi_{\bar{T}}}$ is not product with respect to 
the bipartition $(\bar{T}\cap S_{i_m})|[\bar{T}\setminus (\bar{T}\cap S_{i_m})]$, as otherwise $\ket{\psi_T}\ket{\phi_{\bar{T}}}$ would be
product with respect to unions of the sets $S_k$. 
This, via Lemma \ref{LLem1},
implies that the condition $P^{\mathrm{sym}}_{S_{i_m}}\ket{\psi}=\ket{\psi}$
cannot be satisfied, which, in turn, contradicts (\ref{Ksym}). This completes the proof.
%
%
\end{proof}

Two remarks are in order here. First, if $K=2$, i.e., the state $\rho_{\textbf{A}}$ can be regarded as  having two
subsystems $S$ and $\bar{S}$ defined on the symmetric subspaces of the
corresponding Hilbert spaces, then Eq. (\ref{dekompozycja}) simplifies to 
\begin{equation}
 \rho_{\textbf{A}}=\sum_{i}p_i\varrho^i_S\ot\bar{\varrho}^i_{\bar{S}},
\end{equation}
that is, the state $\rho_{\textbf{A}}$ is simply separable across $S|\bar{S}$. Second, a
straightforward corollary of this fact is a condition for a multipartite state to 
be genuinely entangled.
\begin{cor} \cite{krotka}
If $\rho$ satisfies (\ref{Ksym})
and does not admit the decomposition (\ref{dekompozycja}), then it is genuinely multipartite entangled. In particular, for $K=2$, if $\rho$ is entangled across $S|\bar{S}$, then it is GME.
\end{cor}
This gives a simple method of
checking whether a state is genuinely entangled: once we know that $\rho$ has the symmetry (\ref{Ksym}), one has to check only some of the bipartitions to confirm that $\rho$ is GME. The existence of this symmetry may be given in advance or verified with a direct computation.

We now move on to derive generalizations of the above fact. For this purpose, let us notice that the key property of the symmetric subspaces that enabled us to prove Fact \ref{fact1} is that they only consist of either fully product vectors or multipartite states that are genuinely entangled (see, e.g., Refs. \cite{AnnPhysLew,JaJordiPRA}). 
This observation can be further exploited to extend Fact \ref{fact1} to other types of subspaces.
Precisely, as we demonstrate below, the symmetric subspaces in (\ref{Ksym}) can be replaced by arbitrary genuinely entangled subspaces of $\mathcal{H}_{S_i}$.

For pedagogical reasons we begin with the case of $K=2$ and consider 
a bipartition $S|\bar{S}$ of the Hilbert space $\mathcal{H}_{d,N}=\calH_S \otimes \calH_{\bar{S}}$ ($|S|+|\bar{S}|=N$). Let us then distinguish two genuinely entangled subspaces
of the Hilbert spaces  corresponding to  the groups $S$ and $\bar{S}$, 
$\mathcal{V}_{\mathrm{ges}}\subset \mathcal{H}_{d,|S|}$ and $\mathcal{V}_{\mathrm{ges}}'\subset\mathcal{H}_{d,|\bar{S}|}$,
with the projections $P_{\bar{S}}^{\mathrm{ges}}$ and  $P_{\bar{S}}^{\mathrm{ges}}$, respectively. The following holds true.
%
\begin{fakt}\label{fact2}
 Let $\rho_{\textbf{A}}$ be an $N$-partite state acting on $\mathcal{H}_{d,N}$ such that for
some bipartition $S|\bar{S}$ the condition 
\begin{equation}\label{2GES}
 P_{S}^{\mathrm{ges}}\ot
P_{\bar{S}}^{\mathrm{ges}}\rho_{\textbf{A}} P_{S}^{\mathrm{ges}}\ot
P_{\bar{S}}^{\mathrm{ges}}=\rho_{\textbf{A}}.
\end{equation}
If $\rho_{\textbf{A}}$ is not GME, then it is separable across the bipartition $S|\bar{S}$.
\end{fakt}
\begin{proof}The fact that $\rho_{\textbf{A}}$ is not GME implies that it can be written as a convex 
combination of states that are biseparable across various bipartitions. 
Assume then that in this decomposition there is a state that is biseparable
across a bipartition $T|\bar{T}$ different than $S|\bar{S}$ (and that it is not separable across $S|\bar{S}$). It is not difficult to see that the range of this state contains product vectors $\ket{\psi_{T|\bar{T}}}=\ket{\phi_T}\ket{\varphi_{\bar{T}}}$ which obey the symmetry (\ref{2GES}), meaning that 
\begin{equation}\label{conditions1}
P_{X}^{\mathrm{ges}}
\ket{\psi_{T|\bar{T}}}=\ket{\psi_{T|\bar{T}}} \qquad (X=S,\bar{S}).
\end{equation}
Let us now use the fact that the separability cut of $\ket{\psi_{T|\bar{T}}}$
belongs to $S$ or $\bar{S}$; otherwise these must be the same bipartitions. For concretness we
assume the latter to be $S$. Then, by tracing out the $\bar{S}$ part of 
$\ket{\phi_T}\ket{\varphi_{\bar{T}}}$ we obtain a mixed state $\varrho_S$ acting on the Hilbert space corresponding to the group $S$, which is biseparable. This contradicts 
the fact that, according to (\ref{conditions1}), $P_{S}^{\mathrm{ges}}\varrho_S P_{S}^{\mathrm{ges}}=\varrho_S$, meaning that $\varrho_S$ must be genuinely multipartite entangled.
\end{proof}

We thus obtain another separability condition:
\begin{cor} If a multipartite $\rho_{\textbf{A}}$
is entangled across a bipartition $S|\bar{S}$ and satisfies 
(\ref{2GES}), then it must be GME.
\end{cor}



It is not difficult to realize that Fact \ref{fact2} can also be formulated more generally, just as Fact \ref{fact1}, for any $K$-partition with arbitrary $K$.

Moreover, Fact \ref{fact1} and Fact \ref{fact2} with its discussed extension can be combined into a more general statement, in which  for every subset $S_i$ of the parties, the related subspace of the corresponding Hilbert space can either be genuinely entangled or symmetric. Namely, the following holds.

%
%
%
%
%
%
\begin{fakt}\label{fact4}
 Let $\rho_{\textbf{A}}$ be an $N$-partite state acting on $\mathcal{H}_{d,N}$ such that for
some $K$-partition $S_1|\ldots|S_K$ of the parties the following
condition holds true 
\begin{equation}\label{KGES}
 P_{S_1}\ot\ldots\ot
P_{S_K}\,\rho_{\textbf{A}}\, P_{S_1}\ot\ldots\ot
P_{S_K}=\rho_{\textbf{A}},
\end{equation}
where $P_{S_i}$ stands for a projector onto 
a symmetric or genuinely entangled subspace
of the Hilbert space corresponding to the group $S_i$.

Then, if $\rho_{\textbf{A}}$ is not GME, it can be written as
in Eq. (\ref{dekompozycja}) with $T$ being sums of the sets $S_i$.
\end{fakt}
\begin{proof}
%
%
The proof follows similar lines of reasoning as that of Fact \ref{fact2}. Since $\rho_{\textbf{A}}$ is not GME it admits the decomposition into a convex combination of biseparable states (\ref{dekompozycja}).
Assume then that in this decomposition there is a biseparable state $\rho_{T|\bar{T}}$ for which 
$T$ is \textit{not} a union of some of the sets $S_1,\ldots,S_K$. Clearly, this  $\rho_{T|\bar{T}}$ can be written as the following convex combination 
\begin{equation}
\rho_T=\sum_{i}q_i \proj{\psi_{T}^i}\otimes \proj{\phi_{\bar{T}}^i}
\end{equation}
in which at least one of the pure states 
$\ket{\psi_T^i}\ket{\phi_{\bar{T}}^i}$ is not product with respect to unions of the sets $S_k$.

Let us then consider one such  state and denote it simply by $\ket{\psi_T}\ket{\phi_{\bar{T}}}$. Due to the fact that $T$ is not a union of $S_k$'s there are sets 
$S_{i_1},\ldots,S_{i_l}$ ($l \ge 1$; it is possible that there is only one such state),  which have nonempty overlaps with both $T$ and $\overline{T}$. 

Let us denote by $\mathcal{V}_i$'s subspaces with respective projection $P_{S_i}$'s.
We now need to consider two cases: 
(i) at least one of the subspaces $\mathcal{V}_{i_1},\ldots,\mathcal{V}_{i_l}$ is genuinely entangled, 
(ii) all of them are symmetric.

In the case (i), the fact that one of $\mathcal{V}_{i_1},\ldots,\mathcal{V}_{i_l}$, say $\mathcal{V}_{i_m}$, is genuinely entangled contradicts 
the condition (\ref{KGES}). This is because the marginal density matrix $\rho_{S_{i_m}}$ of 
$\ket{\psi_T}\ket{\phi_{\bar{T}}}$ corresponding to the subset $S_{i_m}$ is
certainly separable across the bipartition $[T\cap S_{i_m}]|[\bar{T}\cap S_{i_m}]$, while 
the condition (\ref{KGES}) implies that $P_{S_{i_m}}\rho_{S_{i_m}}P_{S_{i_m}}=\rho_{S_{i_m}}$, meaning that it must be supported on a genuinely entangled subspace.

As to the case (ii), we notice that among the sets $S_{i_1},\ldots,S_{i_l}$
there must be at least one, call it $S_{i_m}$, for which
either the state $\ket{\psi_T}$ is not product with respect to the bipartition $(T\cap S_{i_m})|[T\setminus (T\cap S_{i_m})]$ or $\ket{\phi_{\bar{T}}}$ is not product with respect to the bipartition $(\bar{T}\cap S_{i_m})|[\bar{T}\setminus (\bar{T}\cap S_{i_m})]$.
This, via Lemma \ref{LLem1},
implies that the condition $P^{\mathrm{sym}}_{S_{i_m}}\ket{\psi_T}\ket{\phi_{\bar{T}}}=\ket{\psi_T}\ket{\phi_{\bar{T}}}$
cannot be satisfied which contradicts (\ref{Ksym}). This completes the proof.
%
%
\end{proof}
%


The above fact, with Facts \ref{fact1} and \ref{fact2} as special cases, is one of the main results of the present paper. It is the key ingredient of the construction of novel 
examples of states that are genuinely entangled but not genuinely nonlocal  presented in the upcoming section.

Notice that we could easily generalize the fact to the case of subspaces which contain either genuinely entangled states or fully product ones. As noted earlier, the symmetric subspace bears this feature, but clearly it is not the unique one possessing this property.

\subsection{The construction} \label{konstrukcja}

Having the above observations at hand, we can now pass to the announced  construction of multipartite genuinely entangled states. 

We begin with the simplest, 
yet probably the most interesting case of bipartite states ($K=2$) extended to $N$-partite ones . 

Consider a bipartite state $\rho_{AB}$ acting on 
$\mathcal{H}_{d,2}$ and a pair of positive (P) trace preserving (TP) maps\footnote{A
linear map $\Lambda:\mathcal{B}(\mathcal{H})\to \mathcal{B}(\mathcal{K})$ is positive if when applied to a positive operator it returns a positive operator.}
\begin{equation}\label{maps22}
\Lambda_A:\mathcal{B}(\mathcal{H}_{d,1})\to \mathcal{B}(\mathcal{V}_1),\qquad \Lambda_B:\mathcal{B}(\mathcal{H}_{d,1})\to \mathcal{B}(\mathcal{V}_2),
\end{equation}
where $\mathcal{V}_1 \subseteq \mathcal{H}_{d,|S|}$ and $\mathcal{V}_2 \subseteq \mathcal{H}_{d',|\bar{S}|}$ 
are some subspaces belonging to $|S|$-partite and $|\bar{S}|$-partite Hilbert spaces of dimensions $d$ and $d'$, respectively, and $|S|+|\bar{S}|=N>2$. In general, the dimensions can be different, for clarity, however, in further parts of the paper we concentrate on the case $d=d'$. 
Of course, one could also consider a more general  case with the local dimensions differing within each Hilbert space, nevertheless, such generalization is
rather straightforward, and it would unnecessarily complicate the
considerations without providing any additional significant insight. We thus do not consider it here.
 We intentionally 
begin with a quite general class of
PTP maps, as this will be useful later, 
however, by imposing further containts on them 
we will recover the class of maps that can be used in our construction.

%

Let us now assume that the simultaneous action of both $\Lambda_A$ and $\Lambda_B$ on the respective subsystems of the state $\rho_{AB}$ results in a positive operator, i.e., $(\Lambda_A\otimes \Lambda_B)(\rho_{AB})\geq 0$, and consider the resulting $N$-partite quantum state 
\begin{equation}\label{StanSigma2}
 \sigma_{\textbf{A}}=\left(\Lambda_A
\ot\Lambda_B\right)(\rho_{AB}).
\end{equation}
%
 Notice here that we allow for the situation that one of the sets $S$ or $\bar{S}$ contains only one element, meaning that the output Hilbert space of this map is single-partite, and the state $\rho_{AB}$ is not expanded on this subsystem. 

Finally, we will need to assume that:  (i) both PTP maps $\Lambda_A$ and $\Lambda_B$
 are invertible, which means  that there exist $\Lambda_A^{-1}:\mathcal{B}(\mathcal{V}_1)\to \mathcal{B}(\calH_{d,1})$ and $\Lambda_B^{-1}:\mathcal{B}(\mathcal{V}_2)\to \mathcal{B}(\calH_{d,1})$ such that $\Lambda^{-1}_X(\Lambda_X(Z))=Z$ for any $Z$ and with $X=A,B$, (ii) both inverses are positive too.
This is a strong assumption as it is known that PTP maps have positive inverses in the above sense iff they are isometric conjugations, i.e., 
$\Lambda_X(\cdot)=V(\cdot)V^{\dagger}$ with $V$ satisfying $V^{\dagger}V=\mathbbm{1}_d$, or transpositions (see Corollary 6.2 of Ref. \cite{Wolf}). In what follows, for obvious reasons, we only focus on the first type of mappings  
and on many occasions call them shortly isometries. 

With all this at hand we can now state our second of the main results, which is a generalization of the one proven in Ref. \cite{krotka}.
\begin{thm}\label{thm1}
Consider a bipartite entangled state $\rho_{AB}$ acting on $\mathcal{H}_{2,d}$ and
two subspaces $\mathcal{V}_1 \subseteq \mathcal{H}_{|S|,d'}$ and $\mathcal{V}_2 \subseteq \mathcal{H}_{|\bar{S}|,d'}$ in some $|S|$ and $|\bar{S}|$-partite Hilbert spaces of local dimension $d'$ such that $|S|+|\bar{S}|=N>2$. Each subspace is assumed to be either symmetric or genuinely entangled in the corresponding Hilbert space. 
Then, the $N$-partite state (\ref{StanSigma2}) with $\Lambda_A$ and $\Lambda_B$ being isometric mappings is GME.
\end{thm}
\begin{proof} Let us assume the contrary, i.e., that $\sigma_{\textbf{A}}$ is not GME. Due to the fact that 
the output subspaces of the positive maps $\Lambda_i$ are either symmetric or genuinely 
entangled, Fact \ref{fact4} tells us that $\sigma$ must be separable across
the bipartition $S|\bar{S}$, i.e., 
\begin{equation}
\sigma_{\textbf{A}}=(\Lambda_A\otimes \Lambda_B)(\rho_{AB})=\sum_{i}q_i \rho_{S}^i\otimes \rho_{\bar{S}}^i.
\end{equation}
Since $\Lambda_A$ and $\Lambda_B$ are invertible, we can express 
the state $\rho_{AB}$ as
\begin{equation}
\rho_{AB}=\sum_{i}q_i \Lambda_A^{-1}(\rho_S^i)\otimes \Lambda_B^{-1}(\rho_{\bar{S}}^i).
\end{equation}
Since $\Lambda^{-1}_X$ $(X=A,B)$ are positive too, this leads to a contradiction with the assumption that $\rho_{AB}$ is entangled.
\end{proof}

It turns out that this statement can be generalized to $K$ different than two, i.e., 
as the initial state we can use $K$-partite states that are genuinely entangled. We have the following.
\begin{thm}\label{thm2}
Let $\rho_{A_1\ldots A_K}$ be a $K$--partite GME state. Consider a collection of isometries:
\beqn
\Lambda: \calB(\calH_{A_i})\rightarrow \calB(\calV_i)\quad (i=1,\cdots, K),
\eeqn
where each $\calV_i$ is assumed to be either a symmetric or a genuinely entangled subspace of the $|S_i|$--partite subspace $\calH_{S_i}=\calH_{|S_i|,d'}=(\mathbb{C}^{d'})^{\otimes |S_i|}$ ($|S_1|+\cdots+ |S_K|=N$) corresponding to the group of parties $S_i$.

Then, the state %
\begin{equation}\label{StanBorys}
\sigma_{\textbf{A}}=(\Lambda_1\otimes\ldots\otimes \Lambda_K)(\rho_{A_1\ldots A_K}).
\end{equation}
is GME.
\end{thm}

%
%
%
%
%
%
\begin{proof} Assume that the resulting state $\sigma_\textbf{A}$ is not GME. Due to the fact that 
the output subspaces of the positive maps $\Lambda_i$ are either symmetric or genuinely 
entangled, Fact \ref{fact4} tells us that $\sigma_\textbf{A}$ must be a convex combination of states that are separable across certain bipartitions $T|\bar{T}$, where $T$ are unions of the sets $S_i$ [see Eq. (\ref{dekompozycja})]. Then, since 
all maps $\Lambda_i$ are isometries, this would mean that the ,,initial'' state $\rho_{A_1\ldots A_K}$ is not GME, which contradicts the assumption of the theorem.
\end{proof}

We thus have a quite general construction of $N$-partite GME states from 
$K$-partite GME ones with $K<N$. Below, we demonstrate how the method works in practice, 
constructing
a few examples of noisy multipartite states that are genuinely entangled.

\subsection{Examples}
\label{Examples}

Let us now illustrate our method by applying it to a few paradigmatic classes of states.

\paragraph{Example 1.} Let us begin with an example considered already in Ref. \cite{krotka} which concerns 
the well-known class of isotropic states \cite{isotropic}:
\begin{equation}\label{isotr}
\rho_{\mathrm{iso}}(p)=p\proj{\phi_{d}^+}+(1-p)\frac{\mathbbm{1}\otimes \mathbbm{1}}{d^2}\qquad
(0\leq p\leq 1),
\end{equation}
where $\ket{\phi_{d}^+}=(1/\sqrt{d})\sum_{i=0}^{d-1}\ket{ii}$
is the maximally entangled state of two qudits, while $\mathbbm{1}$ is a $d\times d$ identity matrix. For the maps $\Lambda_X$ we then take rank-one (i.e. with a single Kraus operator) completely positive maps 
given as $\Lambda_A(\cdot)=V_L(\cdot)V_L^{\dagger}$ and $\Lambda_B(\cdot)=V_{N-L}(\cdot)V_{N-L}^{\dagger}$ with $V_{L}$ being an isometry defined as 
\beq\label{izometria-prosta}
V_L\ket{i}=\ket{i}^{\otimes L},
\eeq 
where $\{\ket{i}\}$ is the computational basis in $\mathbbm{C}^d$. The output subspaces of both channels are certain subspaces of $\mathrm{Sym}((\mathbbm{C}^d)^{\otimes L})$ and $\mathrm{Sym}((\mathbbm{C}^d)^{\otimes (N-L)})$, respectively, and therefore our results can be applied here. 
An application of the isometries (\ref{izometria-prosta})  to $\rho_{\mathrm{iso}}(p)$ results in the following class of $N$-qudit states
\begin{equation}\label{example1}
\rho_{N}(p)=p\proj{\mathrm{GHZ}^{(+)}_{d,N}}+(1-p)\frac{\mathcal{P}_{d,L}\otimes \mathcal{P}_{d,N-L}}{d^2},
\end{equation}
with $\ket{\mathrm{GHZ}^{(+)}_{d,N}}=(1/\sqrt{d})\sum_{i=0}^{d-1} \ket{i}^{\otimes N}$ being the well-known $d$-level GHZ state and $\mathcal{P}_{d,L}=\sum_{i=0}^{d-1}\proj{i}^{\otimes L}$. The states $\rho_N(p)$ are GME for any $p$ for which the isotropic states are entangled, i.e., $p>1/(d+1)$.

We now present a different extension of the isotropic state being an illustration to Theorem \ref{thm1}. With this aim consider the $n$ qubit GES from Ref. \cite{GESy} spanned by the unnormalized vectors:
\beqn \label{wektory-w-ges}
\ket{0}\sum_{k=2}^{n}\ket{(2^{n-k}+j)_2}-\ket{1}\ket{(j)_2}, \quad j=0,1,\dots, 2^{n-2}-1,
\eeqn
where $(\cdot)_2$ is the $(n-1)$ digit binary representation of a number.
Assume that an orthonormal basis for this GES is  $\{\ket{\phi_j}\}_{j=0}^{2^{n-2}-1}$. Set now $d=2^{n-2}$ and consider the following isometries with the output in the GES:
\beqn
V_n \ket{j}=\ket{\phi_j}, \qquad 0\le j \le 2^{n-2}-1.
\eeqn
Applying locally such isometries to $A$ and $B$ of (\ref{isotr}) gives the following $2n$ qubit state:
\beqn
\hat{\rho}_{2n}(p)= p \proj{\Phi^{+}_{\mathrm{GES}}}+(1-p) \frac{\calP_{\mathrm{GES}}\otimes \calP_{\mathrm{GES}}}{2^{2n-4}},
\eeqn
where $ \ket{\Phi^{+}_{\mathrm{GES}}}=(1/\sqrt{2^{n-2}-1})\sum_{j=0}^{2^{n-2}-1}\ket{\phi_j}\ket{\phi_j}$ and $\calP_{\mathrm{GES}}=\sum_{j=0}^{2^{n-2}-1}\proj{\phi_j}$ is the projection on the GES under scrutiny. Reasoning as above, we have that $\hat{\rho}_{2n}(p)$ is GME for $p>1/(2^{n-2}-1)$.

As an example consider the case $n=3$. Then, the GES is spanned by:
\beqn
&&\ket{\phi_0}=\frac{1}{\sqrt{3}} (\ket{001}+\ket{010}-\ket{100}), \\
&&\ket{\phi_1}=\frac{1}{\sqrt{6}}(\frac{3}{2}\ket{001}+\ket{010}-\frac{1}{2}\ket{011}-\frac{3}{2}\ket{100}+\frac{1}{2}\ket{101}).
\eeqn
%

\paragraph{Example 2.} Building on the above example we can construct a more general class of GME states. Consider any pure entangled state of two qudits written in the Schmidt form as
\begin{equation}\label{pure}
\ket{\psi_{\boldsymbol{\mu}}}=\sum_{i=0}^{d-1}\sqrt{\mu_i} \ket{ii}
\end{equation}
with $\boldsymbol{\mu}$ being a vector consisting of the Schmidt coefficients $\mu_i> 0$, and consider its mixture with white noise
\begin{equation}\label{genmixture}
\rho_{\boldsymbol{\mu}}(p)=p \proj{\psi_{\boldsymbol{\mu}}}+(1-p)\frac{\mathbbm{1}\otimes \mathbbm{1}}{d^2}\qquad
(0\leq p\leq 1).
\end{equation}
It is known that this state is entangled iff \cite{VidalTarrach}:
\begin{equation}\label{prawdcond}
p\leq p^{\mathrm{sep}}_{\boldsymbol{\mu}}\equiv  \frac{1}{d^2\theta+1}
\end{equation}
with $\theta=\max_{i\neq j}\{\sqrt{\mu_i\mu_j}\}$. The application of $\Lambda_A$ and $\Lambda_B$ introduced in the previous example to the corresponding subsystems of $\rho_{\boldsymbol{\mu}}(p)$ leads us to the following quite general class of $N$-partite states
\begin{equation}\label{mixture}
\rho_{N,\boldsymbol{\mu}}(p)=p \proj{\psi_{N,\boldsymbol{\mu}}}+(1-p)\frac{\mathcal{P}_{d,L}\otimes \mathcal{P}_{d,N-L}}{d^2}\qquad
(0\leq p\leq 1),
\end{equation}
where $\ket{\psi_{N,\boldsymbol{\mu}}}$ is the so-called $N$-qudit Schmidt state, i.e., 
a generalization of the $N$-qudit GHZ state given by 
\begin{equation}
\ket{\psi_{N,\boldsymbol{\mu}}}=\sum_{i=0}^{d-1}\sqrt{\mu_i}\ket{i}^{\otimes N}.
\end{equation}
We thus obtain a quite general class of $N$-qudit states which are mixtures of the Schmidt states and some particular type of noise whose entanglement is straightforward to characterize
via our results: they are genuinely entangled iff the condition (\ref{prawdcond}) is satisfied.

\paragraph{Example 3.} 
The above choice of the isometry $V_L$ is probably the simplest one that one could think of. Now, our aim is to provide a less direct example, in particular one that maps
bipartite qudit states into multipartite qubit ones.
Let us consider a particular two-qudit pure state (\ref{pure}) of the form 
\begin{equation}
\ket{\varphi}=\frac{1}{\sqrt{\binom{2N}{d-1}}}\sum_{i=0}^{d-1}\sqrt{\binom{N}{i}\binom{N}{d-1-i}}\,\ket{ii},
\end{equation}
and, again, its mixture with white noise as given in Eq. (\ref{mixture}).
Let us then consider the following isometries 
\begin{equation}
V_N\ket{i}=\ket{D_{N,i}} \qquad V'_N\ket{i}=\ket{D_{N,d-1-i}}\qquad (i=0,\ldots,d-1),
\end{equation}
which map the standard basis in $\mathbbm{C}^{d}$ to $N$-qubit symmetric
Dicke states with $N=d-1$. After applying them to $\ket{\varphi}$ one obtains the following 
$N$-qubit pure state
\begin{equation}\label{Raimat}
V_N\otimes V'_N\ket{\varphi}=\frac{1}{\sqrt{\binom{2N}{d-1}}}\sum_{i=0}^{d-1}\sqrt{\binom{N}{i}\binom{N}{d-1-i}}\,\ket{D_{N,i}}\ket{D_{N,d-1-i}},
\end{equation}
which with a bit of algebra can be shown to be simply the $2N$-qubit state with $d-1$ excitations $\ket{D_{2N,d-1}}$; this is because the factors appearing under the sum are normalization factors of the Dicke states $\ket{D_{N,i}}$ and $\ket{D_{N,d-1-i}}$, and thus 
(\ref{Raimat}) is a normalized sum of all $2N$-qubit kets containing $d-1$ ones and 
$d-1$ zeros, which is nothing but the Dicke state $\ket{D_{2N,d-1}}$.
 
When applied to the mixture of $\ket{\varphi}$ and the white noise, these isometries give rise to the following $2N$-qubit noisy Dicke state
given by
\begin{equation}\label{noisyDicke}
\rho_{2N}(p)=p\proj{D_{2N,d-1}}+(1-p)\frac{P^{\mathrm{sym}}_{d-1}\otimes P^{\mathrm{sym}}_{d-1}}{d^2},
\end{equation}
where $P^{\mathrm{sym}}_{d-1}$ stands for the projection onto the symmetric subspace of
the $(d-1)$-qubit Hilbert space. As before we can easily decide on the values of $p$ for which this state is 
genuinely entangled. From Eq. (\ref{prawdcond}) it follows that the state $\rho_{2N}(p)$
is genuinely entangled for $p>1/(d^2\theta_{\mathrm{ev/odd}}+1)$ with 
\begin{equation}
\theta_{\mathrm{ev}}=\frac{1}{\binom{2(d-1)}{d-1}}\binom{d-1}{\lfloor \frac{d-1}{2}\rfloor}
\binom{d-1}{\lceil \frac{d-1}{2}\rceil}
\end{equation}
for even $d$, and
\begin{equation}
\theta_{\mathrm{odd}}=\frac{1}{\binom{2(d-1)}{d-1}}\binom{d-1}{\frac{d-1}{2}}
\sqrt{\binom{d-1}{\frac{d-3}{2}}\binom{d-1}{\frac{d+1}{2}}}
\end{equation}
for odd $d$. 
%
%
%
%

%

Let us notice that with a little bit more effort one can analogously construct 
a mixture of an $N$-qubit Dicke state with an arbitrary number of excitations
$k=0,\ldots,N$ and some noise of the above type above that are genuinely entangled.

\paragraph{Example 4.} For the last example of an application of the method to bipartite states consider another particular two--qubit  Bell diagonal state (the isotropic state is also in this class):
\beqn\label{bell-diag}
\rho_{Bell}(p)= p \proj{\phi^+_2}+(1-p) \proj{\phi^-_2},
\eeqn
where $\ket{\phi^{\pm}_2}=1/\sqrt{2} (\ket{00}\pm \ket{11})$. One can easily verify that the state is entangled iff $p\ne 1/2$. Using the isometry (\ref{izometria-prosta}) with any $L$, we extend this state to the mixture of two GHZ states with opposite relative phases, i.e.
\beqn\label{bell-diag-extended}
\sigma_{Bell}(p)= p \proj{\mathrm{GHZ}^{(+)}_{2,N}}+(1-p) \proj{\mathrm{GHZ}^{(-)}_{2,N}},
\eeqn
where $\ket{\mathrm{GHZ}^{(-)}_{2,N}}=1/\sqrt{2} (\ket{0}^{\otimes N}-\ket{1}^{\otimes N})$. Clearly, the resulting state is entangled across the bipartite cut induced by the isometries iff $p\ne1/2$, and, in consequence this constitutes the region in which it is also GME. Note that the subspace after the extension, which is spanned by $\mathrm{GHZ}^{(+)}$ and  $\mathrm{GHZ}^{(-)}$, falls into both categories: it is symmetric and genuinely entangled.

\paragraph{Example 5.} We now move to examples of multipartite departure states and begin with the case $K=3$.
To that end, we consider the following three-qubit density matrix introduced in Ref. \cite{TothAcin}:
\begin{equation}\label{tripartite}
\rho_{ABC}=\frac{1}{8}\mathbbm{1}_8+\frac{1}{8}\sum_{i=1}^3\left[\frac{1}{3}\mathbbm{1}_2\otimes \sigma_i\otimes\sigma_i-\frac{1}{2}\sigma_i\otimes\mathbbm{1}_2\otimes\sigma_i\right],
\end{equation}
where $\sigma_i$ for $i=1,2,3$ are the standard Pauli matrices. 

Let us then consider the isometries $\Lambda_i=V_{L_i}(\cdot)V_{L_i}^{\dagger}$ $(i=1,2,3)$, where $L_i$ are integers greater than one and $V_L$ is introduced above. When applied to the subsystems of the state $\rho_{ABC}$, they give the following $N$-qubit state
\begin{eqnarray}\label{Ntripartite}
\sigma_{\textbf{A}}&\equiv &\Lambda_1\otimes\Lambda_2\otimes\Lambda_3(\rho_{ABC})\nonumber\\
&=&\frac{1}{8}\mathcal{P}_{L_1,2}\otimes\mathcal{P}_{L_2,2}\otimes \mathcal{P}_{L_3,2}+\frac{1}{8}\sum_{i=1}^{3}\left[\frac{1}{3}\mathcal{P}_{L_1,2}\otimes \Sigma_{L_2,i}\otimes \Sigma_{L_3,i}-\frac{1}{2}\Sigma_{L_1,i}\otimes\mathcal{P}_{L_2,2}\otimes\Sigma_{L_3,i}\right],
\nonumber\\
\end{eqnarray}
where $\mathcal{P}_{L,2}$ are defined above and $\Sigma_{L,i}=V_L \sigma_i V_L^{\dagger}$ with $i=1,2,3$ are Pauli matrices embedded in the $L$-qubit Hilbert space. As discussed in Ref. \cite{TothAcin}, the state $\rho_{ABC}$ is genuinely entangled, and thus by virtue of Fact \ref{fact4}, the state $\sigma_{\textbf{A}}$ is genuinely entangled too.

\paragraph{Example 6.} Here, we give yet another illustration of the method in the case of a multipartite state as a departure one. In particular, we consider a $K$-qubit mixed state with the support in the subspace spanned by the the $W$ state:
\beqn
\ket{W_{2,K}}=\frac{1}{ \sqrt{K}}(\ket{10\dots 00}+\ket{01\dots 00}+\cdots +\ket{00\dots 01}) .
\eeqn
and its ,,complement''
\beqn
\ket{\hat{W}_{2,K}}=\frac{1}{ \sqrt{K}}(\ket{01\dots 11}+\ket{10\dots 11}+\cdots+\ket{11\dots 10})=\sigma_x ^{\otimes K} \ket{W_{2,K}} .
\eeqn
It is known that any incoherent mixture of these states:
\beqn\label{mixture-W-GHZ}
\rho_{A_1A_2\dots A_K}= p \proj{W_{2,K}}+(1-p)\proj{\hat{W}_{2,K}}
\eeqn
is genuinely entangled for any $K$ \cite{kaszlikowski}. Consider now isometries $V: \mathbb{C}^2 \rightarrow (\mathbb{C}^2)^{\otimes L}$ acting as follows:
\beqn
&&\ket{0}_{A_i} \longrightarrow \ket{W_{2,L}}_{\textbf{A}^{(L)}_i},\\
&&\ket{1}_{A_i} \longrightarrow \ket{\hat{W}_{2,L}}_{\textbf{A}^{(L)}_i},
\eeqn
where $\textbf{A}^{(L)}_i:={A_{(i-1)L+1} \dots A_{i L}}$. As in Example 4, this isometry sends to a subspace which is both symmetric and genuinely entangled.
An application of $V^{\otimes K}$ to the state (\ref{mixture-W-GHZ}) results in a genuinely entangled $N$-partite ($N=LK$) state:
\beqn
\sigma_{\textbf{A}}&=&V^{\otimes K} \rho_{A_1A_2\dots A_K} (V^{\dagger})^{\otimes K}= \non
&=& p \proj{(W_{\hat{W}})_{2,N}}_{\textbf{A}}+(1-p) \proj{(\hat{W}_{W})_{2,N}}_{\textbf{A}},
\eeqn
where
\beqn
&&\ket{(W_{\hat{W}})_{2,N}}_{\textbf{A}}=\frac{1}{\sqrt{K}} \sum_{i=1}^K \left[ \ket{\hat{W}_{2,L}}_{\textbf{A}^{(L)}_i}  \otimes\left(   \bigotimes_{\substack{ f=1\\ f\ne i}}^K \ket{W_{2,L}}_{\textbf{A}^{(L)}_f}  \right)  \right],\non
&&\ket{(W_{\hat{W}})_{2,N}}_{\textbf{A}}=\frac{1}{\sqrt{K}} \sum_{i=1}^K \left[ \ket{W_{2,L}}_{\textbf{A}^{(L)}_i}  \otimes \left(   \bigotimes_{\substack{ f=1\\ f\ne i}}^K \ket{\hat{W}_{2,L}}_{\textbf{A}^{(L)}_f}  \right) \right].
\eeqn
\newline

Concluding this section we note an interesting feature of the states constructed within our approach. Namely, their entanglement properties do not depend on the number of particles in the system, i.e., once the starting state is known to be (genuinely) entangled, the resulting one is guaranteed to be GME regardless of the number of parties after the extension. This is in contrast to the situation occurring in case of many important classes of states not covered by the current aproach, e.g., the mixture of the GHZ state and the white noise (see Ref. \cite{guhne-seevinck}).

\section{Application to (un)steerability and (non)locality of GME states}
\label{Sec:GMN}

Importantly, apart from merely providing examples of GME states, our construction
can also be applied to provide further instances of multipartite states
that are not genuinely multiparty nonlocal. We will also apply the method to present unsteerability of GME states.

\subsection{Multipartite nonlocality}
\label{multipartitenonlocality}

Before being able to state our results we need some preparation.
Assume again that the parties share some state $\rho_{\bold{A}}$, but now on their share 
of this state, each party $A_i$ performs a measurement $M_{i}$ with measurement operators
 denoted by $M_{a_i}^{(i)}$, where $a_i$ labels the outcomes. 
Recall that in order to form a quantum measurement
the operators $M_{a_i}^{(i)}$ must be positive and sum up to the
identity $\mathbbm{1}_d$. If these operators are supported on
orthogonal subspaces, i.e.,
$M_{a_i}^{(i)}M^{(i)}_{a_i'}=\delta_{a_ia'_i}M_{a_i}^{(i)}$, we
call the corresponding measurement \textit{projective} (PM).
Otherwise, it is called a \textit{generalized} measurement (GM; also called
POVM).

These measurements on $\rho_{\mathbf{A}}$ give rise to the probability distribution
\begin{eqnarray}
 p(\mathsf{a}|\mathsf{M})&:=&p(a_1,\ldots,a_N|M_1,
\ldots,M_N)\nonumber\\
&=& \tr \left[ \left( M_{a_1}^{(1)}\otimes \cdots\otimes
M_{a_N}^{(N)}\right) \rho_{\mathbf{A}}\right]  ,
\end{eqnarray}
where each $p(\mathsf{a}|\mathsf{M})$
denotes the probability of obtaining
$\mathsf{a}:=a_1,\ldots,a_N$ upon measuring
$\mathsf{M}:=M_1,\ldots,M_N$. 

In full analogy to the notion of $K$-separability
let us then imagine that for any choice of the measurements $M_1,\ldots,M_N$,
the probability distribution $p(\mathsf{a}|\mathsf{M})$
admits the following convex combination
%
%
\begin{equation}\label{model}
p(\mathsf{a}|\mathsf{M})=\sum_{\calS_K\in\mathdollar_K}p_{S_K}
\int_{\Omega_{S_K}}\mathrm{d}\lambda\,\omega_{\calS_K}(\lambda)
p_1(\mathsf{a}_{S_1}|\mathsf{M}_{S_1},\lambda)\ldots
p_K(\mathsf{a}_{S_K}|\mathsf{M}_{S_K},\lambda).
%
\end{equation}
Here the sum goes over all $K$-partitions $\calS_K\in\mathdollar_K$,
$p_{S_K}$ and $\omega_{\calS_K}$ are probability distributions parameterized
by the $K$-partitions $\mathcal{S}_K$'s, and $\Omega_{S_K}$'s are sets over which $\lambda$'s are distributed. Moreover,
$p_{k}(\mathsf{a}_{S_{k}}|\mathsf{M}_{S_{k}})$, called a {\it response function}, is the probability that the
parties belonging to the set $S_k$ obtain results $\mathsf{a}_{S_k}$ upon
measuring $\mathsf{M}_{S_k}$. 

Now, if we do not impose any conditions on the nature of the response function
$p_i(\mathsf{a}_{S_i}|\mathsf{M}_{S_i})$ (quantum, nonsignalling, etc.), except
that they form a proper probability distribution, i.e., they are nonnegative and sum up
to one, Eq. (\ref{model}) gives the definition of \textit{$K$-locality} of the
state $\rho_{\mathbf{A}}$ due to Svetlichny \cite{Svetlichny}.
Analogously to entanglement, if $K=N$ we call $\rho_{\mathbf{A}}$ \textit{fully local},
while if $K=2$ it is said to be \textit{bilocal}.  If we consider a certain $K$--partition  $\calS_K$ we omit the sum over the partitions in Eq. (\ref{model}) and talk about $K$--locality with respect to $\calS_K$.
States that are local are also said to have a {\it  local hidden variable} (LHV) model (of the corresponding type) or simply a local model.
Finally, if $p(\mathsf{a}|\mathsf{M})$
is not bilocal for any choice of the measurements, then we say that
$\rho_{\mathbf{A}}$ is \textit{genuinely multipartite nonlocal (GMN)}.

The intuition behind Svetlichny's definition is that the correlations
that are not fully local might still display locality
if some parties are grouped together, in the
sense that their statistics can be described by global probability distributions
of any nature. In particular, they do not need to be quantum, i.e. obtainable via Born's rule. It should be emphasized that within this approach each
party still has only access to their subsystem of $\rho_{\mathbf{A}}$, or, in other words,
joint measurements of more than one particle are not allowed. However, although at first sight quite natural, the
definition of Svetlichny has been shown to be inconsistent with the operational
interpretation of nonlocality in multipartite systems
\cite{NonlocalityDefinition,NonlocalityDefinition2}. 
%
%
%
%
%
%
One of the ways to fix this problem is to impose that 
all response functions in (\ref{model}) obey the no-signalling principle. This
gives one of the operational definitions of $K$-locality proposed in
\cite{NonlocalityDefinition,NonlocalityDefinition2}: $\rho_{\mathbf{A}}$ is $K$-local if the
for any choice of the measurements the corresponding probability distribution
$p(\mathsf{a}|\mathsf{M})$ admits the form (\ref{model}) in which all response
functions $p_k(\mathsf{a}_{S_k}|\mathsf{M}_{S_k})$ are nonsignalling.


Although, as already said, the Svetlichny definition is not
consistent with the operational interpretation of nonlocality in multipartite
systems, we use it because it allows us to state our results in a general way. 
Still, our construction is also capable of providing  states that 
are GME but not GMN even if the above operational
definition of nonlocality is used.  Later, we provide some illustrative examples.

Let us now get back to our construction and consider a $K$-partite state $\rho_{A_1\ldots A_K}$, but this time we additionally assume that it has a fully local model for generalized measurements. An example of such a state for $K=2$ would be the isotropic states $\rho_{\mathrm{iso}}(p)$ defined in Eq. (\ref{isotr}), or in the general multipartite scenario
the one constructed in Ref. \cite{Hirsch}. 
Consider then some $K$-partition $S_1|\cdots|S_N$  of $N$ parties
and define a collection of $K$ maps
\beqn\label{mapy}
\Lambda_i: \calB(\calH_{A_i})\to \calB(\calH_{S_i}), \qquad 1 \leq i \leq K 
\eeqn
with $\calH_{S_i}=\mathcal{H}_{|S_i|,d'} $
and $|S_1|+\ldots+|S_K|=N>K$.
As previously, we assume these maps to be trace-preserving, however, instead of assuming that they are positive we impose their dual maps
$\Lambda_i^{\dagger}:\calB(\mathcal{H}_{|S_i|,d'})\to \calB(\calH_{A_i})$
to be positive on products of positive operators (PPPO)\footnote{Recall that a linear map $\Lambda:\calB((\mathbbm{C}^d)^{\otimes N})\to \calB(\mathbbm{C}^d)$ is termed positive on products of positive operators if $\Lambda(P_1\otimes \ldots\otimes P_N)\geq 0$ for any sequence of positive operators $P_i$ \cite{PPPO-horodeccy}.}. The operations defined in Eq. (\ref{mapy}) will serve us to extend the state $\rho_{A_1\ldots A_K}$ to an $N$--partite one just as it was in the previous sections.

The following statement was made by Barrett in \cite{Barrett} and later generalized in \cite{krotka}.
\begin{thm}\label{thm3}
Let $\rho_{A_1\ldots A_K}$ be a state 
acting on $\mathcal{H}_{d,K}$ that has a fully local model for 
generalized measurements. Then, for any collection of 
trace-preserving maps $\Lambda_i: \calB(\calH_{A_i})\to \calB(\calH_{S_i})$ whose dual maps $\Lambda_i^{\dagger}$ are 
PPPO and such that  $\otimes_{i=1}^K \Lambda_i(\rho_{A_1\ldots A_K})\geq 0$, the 
$N$-partite state $\sigma_{\textbf{A}}=\otimes_{i=1}^K\Lambda_i(\rho_{A_1\ldots A_K})$
is $K$-local with respect to the $K$-partition $S_1|\ldots|S_K$.
\end{thm}
The proof of this statement can be found in the supplementary material of Ref. \cite{krotka}. Notice that originally it was proven there assuming that $\Lambda_i$ are all quantum channels, i.e. completely positive trace--preserving maps, however, this assumption can be relaxed as we do here, but remembering that we always need to guarantee that $\otimes_{i=1}^K \Lambda_i(\rho_{A_1\ldots A_K})\geq 0$. One also has to bear in mind that the dual map of a trace-preserving positive map is also positive and moreover unital, i.e., preserves the identity.

A combination of this fact with our construction of 
multipartite genuinely entangled states stated in Theorem \ref{thm2}
gives rise to a method of generation of GME states that 
are not GMN. Concretely, we have the following statement.
\begin{thm}\label{thm4}
Let $\rho_{A_1\ldots A_K}$ be a $K$-partite quantum state 
acting on $\mathcal{H}_{d,K}$ that has a fully local model for 
generalized measurements. Let then $\Lambda_i:\calB(\calH_{A_i})\to \calB(\mathcal{V}_i)$
be a collection of isometries with output subspaces $\mathcal{V}_i$ 
being either symmetric or GES. Then, the $N$-partite state 
\begin{equation}
\sigma_{\mathbf{A}}=(\Lambda_1\otimes\ldots\otimes \Lambda_K)(\rho_{A_1\ldots A_K})
\end{equation}
is GME and $K$-local with respect to the $K$--partition induced by  $\Lambda$'s.
\end{thm}
\begin{proof}It is direct to see that the dual map of an isometry is PPPO and therefore it follows from theorem \ref{thm3} that the state $\sigma_{\mathbf{A}}$
has a $K$-local model for generalized measurements with respect to the $K$-partition 
induced by the maps $\Lambda_i$. Then, by virtue of theorem \ref{thm2}, the fact that $\rho_{A_1\ldots A_K}$ is GME implies that so is $\sigma_{\mathbf{A}}$. 
\end{proof}

Now, using the above result, we demonstrate that some of the states introduced in Section \ref{Examples} are GME but not GMN according to the 
Svetlichny definition of nonlocality because they all have bilocal model 
for generalized measurements. Let us consider each example separately from this perspective.

\paragraph{Example 1.} The class of $N$-qudit states from  Eq. (\ref{example1}) was already considered in 
Ref. \cite{krotka} also in the context of nonlocality. For completeness, let us recall these results. 

It was proven in Ref. \cite{TothAlmeidaPRL} that the isotropic states (\ref{isotr})
admit a local model for generalized measurements for any $p$ obeying
\begin{equation}
p\leq p_{\mathrm{GM}}\equiv \frac{(3d-1)(d-1)^{d-1}}{(d+1)d^d}.
\end{equation}
Hence, the $N$-partite states $\rho_{N}(p)$ [eq. (\ref{example1})] are genuinely entangled 
but not genuinely nonlocal for any $1/(d+1)<p\leq p_{\mathrm{GM}}$. 

The same reasoning applies to the states $\hat{\rho}_{2n}(p)$ with the only difference now that $d=2^{n-2}$.

\paragraph{Example 2.} Let us now pass to a more general class of $N$-qudit states
given in Eq. (\ref{mixture}). It was shown in Ref. \cite{TothAlmeidaPRL}
that the two-qudit mixed states (\ref{genmixture}) have a local model for 
generalized measurements for 
\begin{equation}\label{condit2}
p\leq \widetilde{p}_{\mathrm{GM}}\equiv \frac{p_{\mathrm{GM}}}{(1-p_{\mathrm{GM}})(d-1)+1},
\end{equation}
and thus, via our construction, the $N$-partite states are bilocal according to the 
Svetlichny definition of nonlocality for the same values of $p$. So, among all states 
$\rho_{N,\boldsymbol{\mu}}(p)$ defined by Eq. (\ref{mixture}) and parametrized by the vector $\boldsymbol{\mu}$ of Schmidt coefficients, there are GME states 
that are bilocal if 
\begin{equation}\label{cond14}
\widetilde{p}_{\mathrm{GM}}> 1/(d^2\theta_{\boldsymbol{\mu}}+1),
\end{equation}
where, to recall,
$\theta_{\boldsymbol{\mu}}=\max_{i\neq j}\{\sqrt{\mu_i\mu_j}\}$. In fact, it is not difficult to find such $\boldsymbol{\mu}$. In particular, due to the fact that 
$\widetilde{p}_{\mathrm{GM}}\geq (3d-1)/[4d(d+1)]$, one sees that the condition (\ref{cond14})
is satisfied whenever $\theta_{\boldsymbol{\mu}} > (1/d^2)[(4d^2+d+1)/(3d-1)]$ and, clearly, entangled pure states $\ket{\psi_{\boldsymbol{\mu}}}$ for which this last inequality is satisfied do exist. 

Let us finally mention that there are bipartite quantum states admitting local models for generalized measurements for larger values of $p$ than (\ref{condit2}) for which our construction can be applied. As shown in Ref. \cite{TothAlmeidaPRL} the following mixture 
\begin{equation}
\widetilde{\rho}_{\boldsymbol{\mu}}(p)=p\proj{\psi_{\boldsymbol{\mu}}}+(1-p)\rho_A\otimes \frac{\mathbbm{1}_d}{d},
\end{equation}
where $\ket{\psi_{\boldsymbol{\mu}}}$ is a pure state defined in Eq. (\ref{pure})
and $\rho_A$ is its single-party marginal, has a local model for all measurements for
any $p\leq p_{\mathrm{GM}}$.

\paragraph{Example 3.} Unfortunately for the class of states (\ref{noisyDicke})
the condition (\ref{cond14}) is not satisfied and therefore we cannot claim 
his state to be an example of a GME state that is not GMN.

\paragraph{Example 4.} The states (\ref{bell-diag}) violate the CHSH inequality and in consequence are nonlocal in the whole region of entanglement \cite{horo-spin-12}. In turn, the states (\ref{bell-diag-extended}) are not bilocal.

\paragraph{Example 5.} 
As shown in Ref. \cite{TothAcin} the tripartite state 
(\ref{tripartite}) has a hybrid local model, i.e., for PMs on $A$ and POVMs on $B$ and $C$. For this reason the GME state (\ref{Ntripartite}) cannot be claimed to be $3$--local for any kind of measurements. However, if we keep the party $A$ untouched, i.e., map it identically to itself or unitarily conjugate it, we obtain a GME state which has a $3$-local model for projective measurements on $A$ and general measurements for the rest of the parties. This possibility has been already noticed in \cite{LHV-review}

Importantly, the original state can be extended to a qutrit-qubit-qubit state wich does have a local model for general measurements still being genuinely entangled \cite{private-Tulio}. Such a state then, after the application of the extending maps, will be $3$--local for all measurements.
%
%

\paragraph{Example 6.} This state is known to be nonlocal for $p\ne 1/2$. The case $p=1/2$ remains, to the best of our knowledge, unsolved, although the results of Ref. \cite{kaszlikowski} might indicate that for this value of the parameter they are (fully) local.
\newline

%

Consequently, we provide more examples, over those provided in Ref. \cite{krotka}, of $N$-partite quantum states
giving rise to the inequivalence between entanglement and nonlocality in the multiparty case. 

This inequivalence persists if one considers operational definitions of
nonlocality formulated in Refs. \cite{NonlocalityDefinition}. For completeness, we recall the proof of this fact from Ref. \cite{krotka}.
 Consider
those entangled states that have a local model for generalized measurements, meaning that
\begin{equation}\label{modelB}
p_{\rho}(a,b|M_A,M_B)=\int_{\Omega}\mathrm{d}\lambda\,\omega(\lambda)p_{\rho}(a|M_A,\lambda)p_{\rho}(b|M_B,\lambda)
\end{equation}
holds for any $M_A$ and $M_B$ in which one of the response functions, say the one on Bob's side, is quantum, that is, it can be written as $p_{\rho}(b|M_B,\lambda)=\Tr[\sigma_{\lambda}M_{b}]$, where $\sigma_{\lambda}$ is some quantum state representing the hidden variable $\lambda$ and $\{M_b\}$ are the measurement operators of the measurement $M_B$. We have added subscripts to the response functions to stress for which states the model is considered. We then consider two isometries
\begin{equation}\label{LHV}
\Lambda_{A\to A_1}:\calB(\mathbbm{C}^d)\to \calB(\mathbbm{C}^d) \qquad \Lambda_{B\to \bar{S}}:\calB(\mathbbm{C}^d)\to \calB(\calV_2).
\end{equation} 
The first one is trivial as its output space
is the same as the initial one (in fact it is just a unitary conjugation), whereas the send one maps to some subspace 
$\calV_2$ of an $(N-1)$-partite Hilbert space corresponding to parties $A_2,A_2,\cdots, A_N$. Their application to $\rho_{AB}\in \calB(\mathbbm{C}^d \ot \mathbbm{C}^d)$ gives us 
an $N$-partite state $\sigma_{\mathbf{A}}=(\Lambda_{A\to A_1}\otimes \Lambda_{B\to \bar{S}})(\rho_{AB})$ which, according to Theorem \ref{thm4} is GME and has a bilocal model (\ref{model}) with respect to the bipartition $A_1 | A_2 \cdots A_N$ obtained from the local model (\ref{LHV}), which we can explicitly write as
\begin{equation}
p_{\sigma}(\mathsf{a}|\mathsf{M})=
\int_{\Omega}\mathrm{d}\lambda\,\omega(\lambda)
p_{\sigma}(a_1|\hat{M}_{1},\lambda)
p_{\sigma}(a_2,\ldots,a_N|\hat{M}_{2},\ldots,\hat{M}_N,\lambda)
\end{equation}
where the response function of the party $A_1$ is defined as $p_{\sigma}(a_1|\hat{M}_{1},\lambda)=p_{\rho}(a_1|M_A,\lambda)$ with the measurement operators of the measurement $M_A$ given by $M_a=\Lambda^{\dagger}_{A_1\to A}(\hat{M}_a)$, whereas the response function corresponding to the parties $A_2,\ldots,A_N$ is defined as $p_{\sigma}(a_2,\ldots,a_N|\hat{M}_2,\ldots,\hat{M}_N,\lambda)=p(\mathsf{a}|M_B,\lambda)$, where $\mathsf{a}:=a_2,\ldots,a_N$  and the measurement operators of the measurement $M_B=\{M_{\mathsf{a}}\}_{\mathsf{a}}$ are given by
\begin{equation}
M_{\mathsf{a}}=
\Lambda_{\bar{S}\to B}^{\dagger}(M_{a_2}^{(2)}\otimes \ldots\otimes M^{(N)}_{a_N}).
\end{equation}
To show that the model is in agreement with the operational definitions of nonlocality, it suffices to show that the response function $p_{\sigma}(a_2,\ldots,a_N|\hat{M}_2,\ldots,\hat{M}_N,\lambda)$ obeys the no-signaling principle. To this end, we exploit
the fact that Bob's response function in the model (\ref{modelB}) has a quantum realization, meaning that

%
\begin{eqnarray}
p_{\sigma}(a_2,\ldots,a_N|\hat{M}_2,\ldots,\hat{M}_N,\lambda)&\!\!\!=\!\!\!&\Tr[\sigma_{\lambda}M_{\mathsf{a}}]\nonumber\\
&\!\!\!=\!\!\!&\Tr[\sigma_{\lambda}\Lambda_{\bar{S}\to B}^{\dagger}(M_{a_2}^{(2)}\otimes \ldots\otimes M^{(N)}_{a_N})]\nonumber\\
&\!\!\!=\!\!\!&\Tr[\Lambda_{B\to\bar{S}}(\sigma_{\lambda})(M_{a_2}^{(2)}\otimes \ldots\otimes M^{(N)}_{a_N})].
\end{eqnarray}
Due to the fact that $\Lambda_{B\to \bar{S}}$ is a positive trace-preserving map, 
$\Lambda_{B\to\bar{S}}(\sigma_{\lambda})$ is a valid quantum state, and therefore 
$p_{\sigma}(a_2,\ldots,a_N|\hat{M}_2,\ldots,\hat{M}_N,\lambda)$ is a quantum response function, meaning that it necessarily obeys the nonsignaling principle (see, e.g., Ref. \cite{Unified}).

To complete the proof let us mention that there exist local models in which 
one of the response functions is quantum, with the well--known example being the  Barrett model \cite{Barrett} for the Werner states \cite{Werner}. Another example of such a model is the one introduced in Ref. \cite{TothAlmeidaPRL} for the isotropic states (\ref{isotr}). Thus, the $N$-partite states given in Eqs. (\ref{example1}) and (\ref{mixture}) [in the second case provided that the condition (\ref{cond14}) is satisfied]
with $L=1$ (that is, the isometry is applied only to the second subsystem) are examples of GME states that are not GMN according to the operational definitions of nonlocality.

\subsection{Multipartite steering}

Interestingly, this last observation suggests that our method
can also be used to construct GME states that are unsteerable in a sense we will make precise in this section.  This has already been mentioned in \cite{krotka} but no details have been given regarding this issue. We fill this gap here.

We first recall the notion of steering
and  begin with the simplest bipartite scenario as
it serves as the basis for our main constructions.

In a steering scenario two parties share at the beginning an unknown state $\rho_{AB} \in \calB(\mathbbm{C}^{d}\ot \mathbbm{C}^{d})$. The question now is whether one of the parties, say Alice ($A$), by performing some measurements on her share of $\rho_{AB}$, is able to collapse (steer) Bob's reduced state into different ensembles.

More formally, suppose Alice may perform some predetermined number of generalized measurements $M_x=\{M_a^x\}$, where $a$ enumerates results of the $x$--th measurement (for simplicity we assume all measurements to have the same number of outcomes) and $M_a^x$ are the measurement operators, i.e., $M_a^x \ge 0$ and $\sum_a M_a^x= \jedynka_d$ for each $x$. The unnormalized Bob's state after the result $a$ of the $x$--th measurement has been obtained reads
\beq
\stan{a|x}{B}= \tr_{A}\left[( M_a^x \otimes \jedynka_B) \varrho_{AB}\right].
\eeq
The collection $\{\stan{a|x}{B}\}_{a,x}$ is referred to as an {\it assemblage}.
Now, we say that the assemblage $\{\stan{a|x}{B}\}_{a,x}$ is {\it unsteerable} from Alice to Bob,  shortly $A\to B$ unsteerable, if the above can be written as 
\beq\label{unsteerable}
\stan{a|x}{B}=\int_{\Omega}\mathrm{d}\lambda \;\omega (\lambda)p(a|x,\lambda)\sigma_{\lambda}^{B}
\eeq
where $\Omega$ is again the set of hidden variables, 
$\omega(\lambda)$ is some probability density over $\Omega$, 
$p(a|M_A,\lambda)$ is a probability distribution, called as previously the response function of party $A$, and 
$\sigma_{\lambda}^B$ are some quantum states. 

 If the opposite holds we say that the given assemblage is {\it $A\to B$ steerable}.
Further, if {\it all} possible assemblages are unsteerable (to reduce the clutter, when the direction of steering is clear we omit $X \to Y$) we say that the state itself is unsteerable and if there exists a steerable assemblage we call the state steerable.
Unsteerable assemblages (states) are said to have a local hidden state model (LHS) -- the name stemming directly from the form of (\ref{unsteerable}).

Not surprisingly, steering cannot take place when the initial state is separable, i.e., $\rho_{AB}=\sum_{\lambda} p_{\lambda} \varrho_{\lambda}^A\otimes \gamma^B_{\lambda} $. In such case it is enough to take $p(a|x,\lambda):= \tr M_a^x \varrho_{\lambda}^A $ and it is immediate to realize that the resulting $\stan{a|x}{B}$ assumes the form of Eq. (\ref{unsteerable}). Entanglement is thus a necessary condition for steering and the presence of steering certifies entanglement of an unknown state $\rho_{AB}$. Nevertheless, it is not sufficient: although all separable states are unsteerable, some entangled states also have an underlying LHS. These notions are thus inequivalent. In fact, steering as a type of correlations can be placed strictly between entanglement and nonlocality.
As a sidenote, notice that the definition of steering introduces asymmetry into the scenario as it distinguishes the roles of the parties. A priori, there is no guarantee that this asymmetry remains at the fundamental level, i.e., some states are steerable from Alice to Bob [Eq. (\ref{unsteerable}) does not hold] but not in the opposite direction [Eq. (\ref{unsteerable}) with the parties exchanged always holds]. It turns out that it is indeed the case: there are known examples of states which are only one--way steerable \cite{ourPRA}.

In modern approach, steering is considered within the paradigm of (one--sided) device--independent quantum information processing. While it is assumed that Bob's measurements (with the roles of the parties in the formulation we have presented here) are well characterized and trusted, Alice's measurements are treated as black boxes performing unspecified POVMs $M_x$. Bob can perform full state tomography and learn the states $\stan{a|x}{B}$ and later basing on this knowledge verify whether Alice indeed might have steered his share of the state.

Let us now move to the multipartite case, where the situation clearly becomes  more involved. It is not our goal to consider all steering scenarios possible in the multipartite setting (see Ref. \cite{Skrzypczyk-multi}). We only focus  on the case of relevance for future purposes in the present paper, which is the scenario in which a group of $L$ parties perform (untrusted) measurements and the question is whether this might lead to the lack of a LHS model, in the sense of Eq. (\ref{unsteerable}), on the remaining $N-L$ parties.

Let us then consider again  $N$ parties sharing a quantum state
$\rho_{\mathbf{A}}\in B\left((\mathbbm{C}^d)^{\otimes N}\right)$. Let the parties be split into two groups according to some bipartition $S|\bar{S}$. For simplicity and no loss of generality we assume that $S=A_1,\ldots, A_{L}$ and $\bar{S}=A_{L+1},\ldots, A_N$. The number $L$ varies from $1$ to $N-1$ depending on the concrete situation. 

Suppose now that the $i$--th party from $S$ may perform  local measurements $M_{x_i}^{(i)}=\{M_{a_i}^{x_i}\}$ with measurement operators $M_{a_i}^{x_i}$; $x_i$'s enumerate the measurements and the outcomes  are labeled $a_i$. The unnormalized postmeasurement states on $\bar{S}$ now read
\beq
\varrho_{\mathsf{a}_S|\mathsf{x}_S}^{\bar{S}}= \Tr _ {S} \left[ \left( M_{a_1}^{x_1}\otimes \cdots\otimes
M_{a_{L}}^{x_L}\otimes \mathbbm{1}_{\bar{S}}\right) \rho_{\mathbf{A}}\right],
\eeq
where  $\bold{\mathsf{a}}:= a_1 \ldots a_{L}$ and $\mathsf{x}_S:=x_1,\ldots,x_{L}$.

The terminology now is the same as in the bipartite case.  The collection $\{ \varrho_{\mathsf{a}_S|\mathsf{x}_S}^{\bar{S}} \}_{\mathsf{a}_S,\mathsf{x}_S}$ is called the assemblage. It is said to be $S\to \bar{S}$ unsteerable if there is a decomposition of the form
\beq\label{post-multiparty}
\varrho_{\mathsf{a}_S|\mathsf{x}_S}^{\bar{S}}=\int_{\Omega}\mathrm{d}\lambda \;\omega(\lambda) p(\mathsf{a}_S|\mathsf{x}_S,\lambda)\sigma_{\lambda}^{\bar{S}},
\eeq
where $\sigma^{\bar{S}}_{\lambda}$ are some states (called as previously hidden) acting on 
the Hilbert space corresponding to the group $\bar{S}$, whereas
$p(\mathsf{a}_S|\mathsf{x}_S,\lambda)$ is a response function
corresponding to parties belonging to $S$
; $\Omega$ is the set of hidden variables and $\omega(\lambda)$ a probability density over $\Omega$. It is worth noting at this point that the hidden states $\sigma_{\lambda}^{\bar{S}}$ are arbitrary in the sense that no separability condition is imposed on them, although the parties from $\bar{S}$ are spatially separated. If there is no decomposition (\ref{post-multiparty}) for the assemblage, it is called $S\to \bar{S}$ steerable.   If for some choice of measurements the resulting assemblage is $S\to \bar{S}$  steerable the state $\rho_{\boldsymbol{\mathsf{A}}}$ is called $S\to \bar{S}$ steerable. Otherwise, it is said to be $S \to \bar{S}$ unsteerable.

Let us now demonstrate that the methods from Section \ref{konstrukcja} may be directly linked with a construction of genuinely multipartite entangled states which are unsteerable for a given bipartition. The result relies on the observation made by some of us 
in Ref. \cite{ourPRA}. We recall it here in a version adapted to our considerations. Let us take, analogously to Sec. \ref{multipartitenonlocality}, a state $\rho_{AB}$
and two trace-preserving maps 
$\Lambda_{A\to S}$
and $\Lambda_{B\to \bar{S}}$ and assume $\Lambda_{A\to S}$ to have the dual \footnote{Given a linear map $\Lambda:\mathcal{B}(\mathcal{H})\to \mathcal{B}(\mathcal{K})$ its dual is defined to be
a linear map $\Lambda^{\dagger}:\mathcal{B}(\mathcal{K})\to
\mathcal{B}(\mathcal{H})$ satisfying
$\Tr[X\Lambda(Y)]=\Tr[\Lambda^{\dagger}(X)Y]$ for any $X\in
\mathcal{B}(\mathcal{K})$ and $Y\in \mathcal{B}(\mathcal{H})$. Recall that if
$\Lambda$ is positive (in particular completely positive), its
dual $\Lambda^{\dagger}$ is also positive. Moreover, if $\Lambda$
is trace-preserving, i.e., $\Tr[\Lambda(X)]=\Tr X$ for any $X$,
the dual map $\Lambda^{\dagger}$ is unital, i.e.,
$\Lambda^{\dagger}(\mathbbm{1}_{\mathcal{K}})=\mathbbm{1}_{\mathcal{H}}$
with $\mathbbm{1}_{\mathcal{X}}$ being the identity operator
acting on $\mathcal{X}$.}  map PPPO and 
$\Lambda_{B\to \bar{S}}$ to be positive. Assuming then that 
$(\Lambda_{A\to A_1}\otimes \Lambda_{B\to \bar{S}})(\rho_{AB})\geq 0$, we have the following fact (see also Ref. \cite{ourPRA}). 
\begin{fakt}\label{obserwacja-steering}
Let $\rho_{AB}\in \mathcal{B}(\mathcal{H}_{2,d})$ be $A\to B$ unsteerable. Then, for any pair of trace-preserving positive maps 
$\Lambda_{A\to S}$ and $\Lambda_{B\to \bar{S}}$ such that $\Lambda_{B\to \bar{S}}$ has the dual map PPPO and
$(\Lambda_{A\to S}\otimes \Lambda_{B\to \bar{S}})(\rho_{AB})\geq 0$, the $N$-partite state
\begin{equation}\label{stan-sigma}
\sigma_{\mathbf{A}}=(\Lambda_{A\to
S}\ot\Lambda_{B\to \bar{S}})(\rho_{AB})
\end{equation}
is $S \to \bar{S}$ unsteerable.
\end{fakt}
The proof of this fact can be found in Ref. \cite{ourPRA}, however, for completeness we also provide it in the Appendix. We just mention here the origin of the requirement of the dual of  $\Lambda_{B\to \bar{S}}$ to be PPPO. This stems from the fact that an application of this map to the products of local measurement operators on $\bar{S}$ must result in  positive operators, as the latter constitute a generalized measurement.

To construct examples multipartite states that are genuinely entangled and at the same time $S\to \bar{S}$ unsteerable, 
let now 
\beqn \label{izometrie}
&&\Lambda_{A\to S}: B(\mathbbm{C}^d)\to B(\calV_{S}), \non
&&\Lambda_{B\to \bar{S}}:B(\mathbbm{C}^d)\to B(\calV_{\bar{S}})
\eeqn
be two isometries
with $\calV_{S}$ and $\calV_{\bar{S}}$ being subspaces in $\mathcal{H}_{d',L}$ and 
$\mathcal{H}_{d',N-L}$, respectively,
that are either symmetric or genuinely entangled. Then, we have the following fact.
\begin{fakt}
Let $\rho_{AB}$ be entangled but $A\to B$ unsteerable. For any 
pair of the isometries (\ref{izometrie}) the $N$-partite state 
$\sigma_{A}=[\Lambda_{A\to S}\otimes \Lambda_{B\to \bar{S}}](\rho_{AB})$ is GME and $S\to \bar{S}$ unsteerable.    
\end{fakt}
\begin{proof}From Fact \ref{obserwacja-steering} we have that 
the state $\sigma_{\mathbf{A}}$ is $S\to \bar{S}$ unsteerable, whereas from previous facts that it is GME.
\end{proof}

Thus, some of the multipartite states introduced in Section \ref{Examples}
are not only $S|\bar{S}$ local but also $S \to \bar{S}$ unsteerable. 

To have a definition of unsteerability in the multipartite case \'a la definition of locality of Refs. \cite{NonlocalityDefinition}, one 
needs to assume that the response function $p(\mathsf{a}_S|\mathsf{x}_S,\lambda)$
is nonsignaling, and accordingly assume that $L=1$ in the above construction.

\section{Conclusion} 
\label{conclusion}

In Ref. \cite{krotka} we outlined a method of generating 
genuinely entangled $N$-partite states from $K$-partite genuinely entangled ones 
with $K<N$. The aim of the present work is to describe this method in  greater detail and, more importantly, to significantly generalize it. In order to achieve this goal we provide quite general entanglement criteria allowing one to check whether a given multiparty state is genuinely entangled. In the particular case of $K=2$ these criteria are simple to formulate: if
an $N$-partite state is entangled across certain bipartition $S|\bar{S}$ and its
parts corresponding to $S$ and $\bar{S}$ are supported on some genuinely entangled subspace or the symmetric one, then this state is genuinely entangled.

We then apply our method to some known classes of bipartite and multipartite entangled states  obtaining examples of mixed $N$-partite states that are genuinely entangled. At the same time, also generalizing the results of Ref. \cite{krotka}, 
we demonstrate that with the aid of our method we can construct further examples of 
$N$-partite genuinely entangled states which are not genuinely nonlocal. These novel classes of
states provide further support for the statetement made in Ref. \cite{krotka} (see also Ref. \cite{Hirsch}) that entanglement and nonlocality are inequivalent notions in the multiparty case.

Our research provokes several natural questions. For instance, 
it would be interesting to see whether with the aid of our method
one can create further interesting classes of multiparty states
or whether our entanglement criteria could be used to prove that some existing 
classes of states are genuinely entangled. On the other hand, it is of interest to 
see whether the classes of genuinely entangled states we construct admit 
more restrictive local models, in particular, the fully local one.

\section*{Acknowledgments}
This project has received funding from the European Union's Horizon 2020 research and innovation programme under the Marie Sk{\l}odowska-Curie grant agreements No 748549 and No 705109.

\appendix

\section{Proofs}

Here we provide the proof of Lemma \ref{LLem1} from the main text used in the proof of Facts \ref{fact4} and \ref{obserwacja-steering}.

\noindent\textbf{Lemma 1.} \cite{krotka}{\it 
 Let $\ket{\psi}\in\mathcal{H}_{d,N}$ be a pure state product with respect to
some bipartition $T|\bar{T}$. If $P_{S}^{\mathrm{sym}}\ket{\psi}=\ket{\psi}$
where $S$ is a subset of $\textbf{A}$ having  nontrivial 
overlaps with $T$ and $\bar{T}$, i.e.,  $S\cap T \ne \emptyset$ and  $S\cap \bar{T} \ne \emptyset$, then $\ket{\psi}$ is also product with respect
to the bipartition $S|\bar{S}$.
}
%

%
\begin{proof}From the assumption that $\ket{\psi}$ is product with respect to
$T|\bar{T}$ it follows that $\ket{\psi}=\ket{\psi_T}\ket{\phi_{\bar{T}}}$.
Let us then consider two cases: (i) either $T$ or $\bar{T}$ is contained in
$S$, (ii) none them is contained in $S$.

\textit{Case (i).} For concreteness, but without any loss of generality, we
assume that $T$ is contained in $S$. Then, we write the vector $\ket{\phi_{\bar{T}}}$ 
in its Schmidt decomposition with respect to the bipartition $[\bar{T}\cap S]|\bar{T}'$ with $\bar{T}'=\bar{T}\setminus (\bar{T}\cap S)$ as
\begin{equation}\label{Geneva-1}
 \ket{\phi_{\bar{T}}}=\sum_{i}\sqrt{\lambda_i}\ket{\phi'^i_{\bar{T}\cap
S}}\ket{\omega^i_{\bar{T}'}},
\end{equation}
where $\lambda_i$ are the Schmidt coefficients, and $\ket{\phi'^i_{\bar{T}\cap
S}}$ and $\ket{\omega^i_{\bar{T}'}}$ are some orthonormal bases defined on 
$\bar{T}\cap S$ and $\bar{T}'$, respectively. The condition
$P_S^{\mathrm{sym}}\ket{\psi}=\ket{\psi}$ implies then that for any $i$, the
following identity
\begin{equation}
 P_S^{\mathrm{sym}}\ket{\psi_T}\ket{\phi'^i_{\bar{T}\cap
S}}=\ket{\psi_T}\ket{\phi'^i_{\bar{T}\cap S}}
\end{equation}
is satisfied. This, by virtue of the results of Refs. \cite{AnnPhysLew,JaJordiPRA} implies that all vectors $\ket{\psi_T}\ket{\phi'^i_{\bar{T}\cap S}}$ are fully product, i.e., they can be written as
\begin{equation}\label{wielka-dupa}
 \ket{\psi_T}\ket{\phi'^i_{\bar{T}\cap S}}=\ket{e_i}^{\ot|S|}
\end{equation}
with $\ket{e_i}$ being some qudit vectors. In fact, the vectors $\ket{e_i}$ can differ only by a
phase. To see that it is enough to trace out the $\bar{T}\cap S$ part of
(\ref{wielka-dupa}), which gives $\proj{\psi_T}=\proj{e_i}^{\ot |T|}$ for any $i$,
meaning that they are all equal up to
some phase (there are clearly some relations between these phases, but they
are not important for what follows). As a result, 
\begin{equation}\label{dupa}
 \ket{\psi_T}\ket{\phi'^i_{\bar{T}\cap
S}}=\mathrm{e}^{\mathrm{i}\delta_i}\ket{e}^{\ot|S|}.
\end{equation}
Substituting (\ref{wielka-dupa}) into (\ref{Geneva-1}) one finally sees that $\ket{\psi}$ 
is product with respect to $S|\bar{S}$. In fact, it is even fully product on
the subspace $S$.

\textit{Case (ii).} Let us expand both vectors $\ket{\psi_T}$ and
$\ket{\phi_{\bar{T}}}$ as
\begin{equation}
 \ket{\psi_T}=\sum_{i}\sqrt{\lambda_i}\ket{\psi'^i_{T\cap S}}\ket{\omega'^i_{T'}}
\end{equation}
and
\begin{equation}\label{Geneva}
 \ket{\phi_{\bar{T}}}=\sum_{i}\sqrt{\gamma_i}\ket{\phi'^i_{\bar{T}\cap
S}}\ket{\omega''^i_{\bar{T}'}},
\end{equation}
where $\ket{\omega'^i_{T'}}$ and $\ket{\omega''^i_{\bar{T}'}}$ are some
orthonormal bases defined on $T\setminus(T\cap S)$ and $\bar{T}\setminus
(\bar{T}\cap S)$. The condition $P_S^{\mathrm{sym}}\ket{\psi}=\ket{\psi}$
implies that for any pair $i,j$ one has
\begin{equation}\label{equations}
 P_S^{\mathrm{sym}}\ket{\psi'^i_{T\cap S}}\ket{\phi'^j_{\bar{T}\cap
S}}=\ket{\psi'^i_{T\cap S}}\ket{\phi'^j_{\bar{T}\cap
S}}.
\end{equation}
(Notice that $(T\cap S)\cup (\bar{T}\cap S)=S$). Again, due to the results of \cite{AnnPhysLew,JaJordiPRA}
every vector in the above must be fully product. Moreover, one can check that
these fully product vectors may differ only by a phase, so that
\begin{equation}\label{Genf}
 \ket{\psi'^i_{T\cap S}}\ket{\phi'^j_{\bar{T}\cap
S}}=\mathrm{e}^{\mathrm{i}\delta_{ij}}\ket{e}^{\ot |S|},
\end{equation}
where $\delta_{ij}$ are some phases (notice that (\ref{equations}) impose
some conditions on these phases, however, they are not important for the
proof; still they must be such that the state remains product with respect to
$T|\bar{T}$). Putting (\ref{Genf}) into $\ket{\psi_T}\ket{\phi_{\bar{T}}}$ one
finally obtains that $\ket{\psi}$ is product with respect to $S|\bar{S}$. In
particular it is of them form
%
$ \ket{\psi_T}\ket{\phi_{\bar{T}}}=\ket{e}^{\ot |S|}\ket{\bar{\omega}_{\bar{S}}}$,
%
where $\ket{\bar{\omega}_{\bar{S}}}$ is a state defined on $\bar{S}=T' \cup
\bar{T}'$. 
%
%
%
%
This completes the proof.
\end{proof}

Let us now provide the proof of Fact \ref{obserwacja-steering}.

\noindent\textbf{Fact 4.} {\it 
Let $\rho_{AB}\in \mathcal{B}(\mathcal{H}_{2,d})$ be $A\to B$ unsteerable. Then, for any pair of trace-preserving positive maps 
$\Lambda_{A\to S}$ and $\Lambda_{B\to \bar{S}}$ such that $\Lambda_{B\to \bar{S}}$ has the dual PPPO and
$[\Lambda_{A\to S}\otimes \Lambda_{B\to \bar{S}}](\rho_{AB})\geq 0$, the $N$-partite state
\begin{equation}\label{stany-sigma}
\sigma_{\mathbf{A}}=[\Lambda_{A\to
S}\ot\Lambda_{B\to \bar{S}}](\rho_{AB})
\end{equation}
is $S \to \bar{S}$ unsteerable.
}

\begin{proof}
Let $\Lambda^{\dagger}_{S\to A}$ be the dual map of $\Lambda_{A\to S}$ 
%
and define the following operators
\begin{equation}\label{operatory-nowe-S}
 \widehat{M}^{A}_{a_1}=\Lambda^{\dagger}_{S\to A}\big( M_{a_1}^{(1)}\otimes \ldots\otimes 
 M_{a_L}^{(L)}\big)
\end{equation}
%
From the duality of the map in the construction above it immediately follows that the resulting operators form proper POVMs.
 Unnormalized states after the measurements results $a_1$ have been obtained upon measuring $M_1$ by the party $A_1$ on $\sigma_{\boldsymbol{\mathbf{A}}}$ [Eq. (\ref{stany-sigma})] read
\begin{eqnarray}
\varrho_{\mathsf{a}_{S}|\mathsf{M}_S}^{\bar{S}}&=& \Tr_{S} \left[ \left( M_{a_1}^{(1)}\otimes\ldots\otimes M_{a_L}^{(L)}\otimes\mathbbm{1}_{\bar{S}}\right) \sigma_{\mathbf{A}}\right]\nonumber\\
&=& \Tr _ {S} \left[ \left( M_{a_1}^{(1)}\otimes\ldots\otimes M_{a_L}^{(L)}\otimes \mathbbm{1}_{\bar{S}}\right) (\Lambda_{A\to
S}\ot\Lambda_{B\to \bar{S}})(\rho_{AB})\right]\nonumber\\
&=& \Lambda_{B\to \bar{S}} \left( \Tr _ {A} \left[ \widehat{M}^{A}_{\mathsf{a}} \rho_{AB}\right]  \right)\nonumber \\
&=& \Lambda_{B\to \bar{S}} \left( \sum_{\lambda}\omega(\lambda) p(\mathsf{a}_S|\widehat{M}_A,\lambda) \sigma_{\lambda}^B    \right)\nonumber\\
&=& \sum_{\lambda}\omega(\lambda) p(\mathsf{a}_S|\widehat{M}_A,\lambda) \Lambda_{B\to \bar{S}} (\sigma_{\lambda}^B).
\end{eqnarray}
The penulitmate equality, in which $p(a_1|\widehat{M}_A,\lambda)$ is some measurement-dependent probability distribution, is a consequence of $A \to B$ unsteerability of $\rho_{AB}$. Post-measurement states admit thus a decomposition of the form (\ref{post-multiparty}), with the hidden states $\{\Lambda_{B\to \bar{S}} (\sigma_{\lambda}^B)\}_{\lambda}$, so the state $\sigma_{\boldsymbol{\mathbf{A}}}$ is $A_1 \to \bar{S}$ unbisteerable.
\end{proof}

\end{document}